\documentclass[11pt]{article}  

\usepackage[all]{xy}
\usepackage[pdf]{pstricks}
\usepackage{pst-all}
\usepackage{pstricks-add}

\usepackage{graphicx}

\usepackage{caption}
\usepackage{subcaption}

\usepackage{amssymb}
\usepackage{latexsym}
\usepackage[english]{babel}
\usepackage{amsmath}
\usepackage{makeidx}
\usepackage[utf8]{inputenc}
\usepackage{verbatim}
\usepackage{cite}
\usepackage[pagewise]{lineno}
\usepackage{amsthm}
\usepackage{tikz-cd} 
\usepackage{stmaryrd} 
\usepackage[unicode=true,backref]{hyperref}

\newtheorem{theorem}{Theorem}[section]

\newtheorem{proposition}[theorem]{Proposition}

\theoremstyle{definition}
\newtheorem{definition}[theorem]{Definition}

\newtheorem{example}{Example}

\newcommand{\R}{\ensuremath{\mathbb{R}}}

\newcommand{\F}{\ensuremath{\mathbb{F}}}

\begin{document}
	\title{Contact geometry for simple thermodynamical systems with friction}
	
	\author{
		{\bf\large Alexandre Anahory Simoes}\hspace{2mm}
		\vspace{1mm}\\
		{\it\small Instituto de Ciencias Matematicas (CSIC-UAM-UC3M-UCM)}\\
		{\it\small Calle Nicolas Cabrera, 13-15, Campus Cantoblanco, UAM}, {\it\small 28049 Madrid, Spain}\\
		\vspace{2mm}\\
		{\bf\large Manuel de León}\hspace{2mm}
	\vspace{1mm}\\
	{\it\small Instituto de Ciencias Matematicas and Real Academia Española de Ciencias}\\
	{\it\small Calle Nicolas Cabrera, 13-15, Campus Cantoblanco, UAM}, {\it\small 28049 Madrid, Spain}\\
		\vspace{2mm}\\
		{\bf\large Manuel Lainz Valcázar}\hspace{2mm}
		\vspace{1mm}\\
		{\it\small Instituto de Ciencias Matematicas (CSIC-UAM-UC3M-UCM) }\\
		{\it\small Calle Nicolas Cabrera, 13-15, Campus Cantoblanco, UAM}, {\it\small 28049 Madrid, Spain}\\
		\vspace{2mm}\\
		{\bf\large David Martín de Diego}\hspace{2mm}
		\vspace{1mm}\\
		{\it\small Instituto de Ciencias Matematicas (CSIC-UAM-UC3M-UCM) }\\
		{\it\small Calle Nicolas Cabrera, 13-15, Campus Cantoblanco, UAM}, {\it\small 28049 Madrid, Spain}\\
	}

	\maketitle
	
	\begin{abstract}
	Using contact geometry we give a new characterization of a simple but  important class of thermodynamical systems which naturally satisfy the first law of thermodynamics (total energy preservation) and the second law (increase of entropy).  We completely clarify its qualitative dynamics, the underlying geometrical structures and we  show how to use discrete gradient methods.
	\end{abstract} 
	
	\let\thefootnote\relax\footnote{\noindent AMS {\it Mathematics Subject Classification (2010)}. Primary 37J55; Secondary  37D35, 70G45, 80M25.}
	
	\section{Introduction}
	
	In this paper, we introduce a differential geometric framework  that incorporates in a very natural way fundamental thermodynamical concepts as the free energy and the rate of entropy production. 
	
	Typically, in the previous literature, this description needs to introduce appropriate Poisson and dissipation brackets with combined properties that allows the two laws of thermodynamics to be satisfied.

	One of the most successful methods are based on the introduction of \textit{metriplectic structures} (see \cite{Kaufman, morrison} coupling a Poisson and a gradient structure, where the entropy now $S$ is constructed from a Casimir function of the Poisson structure. Other approaches like in  \cite{Ed-Ber, Ed-Ber-2} use similar techniques, called \textit{single generation formalism} introducing a generalized bracket which is naturally divided into two parts: a non-canonical Poisson bracket and a new dissipation bracket. The derived structures are capable of reproducing both reversible and irreversible evolutions providing a unifying formalism for many systems  ruled by the laws of thermodynamics (see also \cite{vander}).	These approaches have proved to be very useful for the description of complex thermodynamical systems and also facilitate their numerical integration. 
	
	Also recently, Gay-Balmaz and Yoshimura \cite{Gay-Balmaz2017, Gay-Balmaz2019} have introduced a  ``variational principle" for the description of thermodynamical systems. Their formulation extends the Hamilton principle of classical mechanics to include irreversible processes by introducing additional phenomenological and variational constraints. 
	
	A  more geometrical approach is based on the use  of contact geometry \cite{Godbillon1969, marle}. In this approach it is proposed that the thermodynamical phase space is equipped with a contact structure. For each function $f$, using the contact structure, it is possible to associate a Hamiltonian vector field $X_f$ which is  the infinitesimal generator of a contact transformation (see Section \ref{sec2}).  In this framework the manifold of equilibrium states is represented by a Legendre submanifold.  The Hamiltonian vector field $X_f$ is tangent to the Legendrian submanifold if and only if the function vanishes on the Legendre submanifold, that is, the Legendre submanifold is contained on the zero level set of the Hamiltonian vector field. The flow of $X_f$ restricted to the Legendrian submanifold are interpreted as thermodynamical processes \cite{Mruga,Mrugala1991,GrPa}. More recently, there has been a resurgence of interest in the study of contact dynamics mainly for the study of systems with dissipation and their geometric properties (\cite{Bravetti2017,Bravetti2018, deLeon2018}). 
	
	In this paper, based on the two laws of thermodynamics and the contact geometry, we study, in Sections \ref{sec2} and \ref{sec3},  the thermodynamical evolution in terms of a different vector field from the Hamiltonian field associated with the structure of contact and  a function. In this case, we study the dynamics associated to the evolution or horizontal  vector field. This vector field is defined in terms of the bi-vector canonically associated with the contact structure. We will check that this vector field satisfies for natural Hamiltonian functions the two laws of thermodynamics and we study its qualitative behaviour. Moreover, the relation with the single generation formalism  is stated without the use  of any artificial construction. Finally, in Section \ref{sec4}, since the evolution vector field is associated to a bi-vector field we analyse the possibility of numerically approaching the flow using discrete gradient methods (see for instance \cite{Gonz,QT1996,ITOH}).

	\section{ Contact geometry}\label{sec2}
	In this section, we consider some ingredients of contact geometry that we will need in the sequel \cite{Godbillon1969, marle,deLeon2018}. 
	
	Let $M$ be a differentiable manifold of dimension $2n+1$ and a 1-form $\eta$ on $M$. We say that $\eta$ is a contact 1-form if $\eta\wedge (d\eta)^n\not=0$ at every point. We say that $(M, \eta)$ is a contact manifold. A distinguished vector field for a contact manifold is the Reeb vector field $R\in {\mathfrak X}(M)$ univocally characterized by 
	\[
	i_R\eta=1 \quad \hbox{and}\quad  i_Rd\eta=0\; .
	\]
	We can define also an isomorphism of $C^{\infty}(M, \R)$ modules by
		\[
	\begin{array}{rrcl}
\flat:& {\mathfrak X}(M)&\longrightarrow& \Omega^1(M)\\
	& X&\longmapsto& i_Xd\eta +\eta(X)\eta
	\end{array}
	\]
	Observe that $\flat^{-1}(\eta)=R$. 
	
	Using the generalized Darboux theorem,  we have canonical coordinates $(q^i, p_i, S)$, $1\leq i\leq n$ in a neighborhooh of every point $x\in M$, such that the contact 1-form $\eta$ and the Reeb vector field  are: 
	\[
	\eta=dS-p_i\; dq^i \qquad \hbox{and} \qquad R=\frac{\partial}{\partial S}\; .
	\]
	
Define the bi-vector $\Lambda$ on $M$ by  
\begin{equation}\label{Lambda:intrinsic}
	\Lambda(\alpha, \beta)=-d\eta(\flat^{-1}(\alpha), \flat^{-1}(\beta)), \qquad \alpha, \beta \in \Omega^1(M)\; .
\end{equation}
In canonical coordinates, 
\begin{equation}\label{Lambda:coordinates}
	\Lambda=\frac{\partial}{\partial p_i}\wedge \left(\frac{\partial}{\partial q^i}+p_i\frac{\partial}{\partial S}\right)
\end{equation}

Define the $C^{\infty}(M, \R)$-linear mapping $$\sharp_{\Lambda}: \Omega^1(M)\rightarrow {\mathfrak X}(M)$$ by $\langle \beta, \sharp(\alpha)\rangle=\Lambda(\alpha, \beta)$ with  $\alpha, \beta \in \Omega^1(M)$. 

Given a function $f\in C^{\infty}(M, \R)$ we will define the following vector fields
\begin{itemize}
\item {\bf Hamiltonian or contact vector field} $X_f$ defined by
\[
X_f=\sharp_{\Lambda} (df)-fR
\]
or in other terms, $X_f$ is the unique vector field such that
\[
\flat(X_f)=df-(R(f)+f)\, \eta\; .
\]

In canonical coordinates:
\[
X_f=\frac{\partial f}{\partial p_i}\frac{\partial}{\partial q^i}
-\left(\frac{\partial f}{\partial q^i}+p_i\frac{\partial f}{\partial S}\right)
\frac{\partial}{\partial p_i}+
\left(p_i\frac{\partial f}{\partial p_i}-f\right)\frac{\partial}{\partial S}
\]

\item {\bf The evolution or horizontal  vector field}
\[
{\mathcal E}_f=\sharp_{\Lambda} (df)=X_f+fR
\]
or
\[
\flat({\mathcal E}_f)=df-R(f)\, \eta\; .
\]
In canonical coordinates:
\[
{\mathcal E}_f=\frac{\partial f}{\partial p_i}\frac{\partial}{\partial q^i}
-\left(\frac{\partial f}{\partial q^i}+p_i\frac{\partial f}{\partial S}\right)
\frac{\partial}{\partial p_i}+
p_i\frac{\partial f}{\partial p_i}\frac{\partial}{\partial S}
\]
We will see in the next section that the evolution vector field will be useful to describe some simple thermodynamical systems with friction where the variable $S$ will play the role of the entropy of the system. 
\end{itemize}

The pair $(\Lambda, E=-R)$ is a particular case of   Jacobi structure since it satisfies
\[
[\Lambda, \Lambda]=2E\wedge \Lambda \quad \hbox{and} \quad [\Lambda, E]=0\; .
\]	
	From this  Jacobi structure we can define the Jacobi bracket as follows: 
	\[
	\{f, g\}=\Lambda(df, dg)+f E(g)-g E(f), \quad f, g\in C^{\infty}(M, \R)
	\]
	The mapping  $\{\;  ,\;  \}: C^{\infty}(M, \R) \times  C^{\infty}(M, \R) \longrightarrow C^{\infty}(M, \R)$ is bilinear, skew-symmetric and satisfies the Jacobi’s identity but, in general, it does not satisfy the Leibniz rule; this last property is replaced by a weaker condition: 
	\[
	\hbox{Supp} \ \{f, g\}\subset 	\hbox{Supp} \ f\cap \hbox{Supp} \ g\; .
	\]
	In this sense, this bracket generalizes the well-known Poisson brackets. Indeed, a Poisson manifold is a particular case of Jacobi manifold. 
	
	In local coordinates
	\begin{eqnarray*}
	\{f, g\}&=&\frac{\partial f}{\partial p_i}\frac{\partial g}{\partial q^i}-
	\frac{\partial f}{\partial q^i}\frac{\partial g}{\partial p_i}
	-\frac{\partial f}{\partial S}\left(p_i\frac{\partial g}{\partial p_i}-g\right)+\frac{\partial g}{\partial S}\left(p_i\frac{\partial f}{\partial p_i}-f\right)
	\end{eqnarray*}

	It is also interesting for us to introduce the bracket (\textit{Cartan bracket}) that now does not obey the Jacobi identity
	\begin{eqnarray*}
		[f, g]&=&\Lambda (df, dg)\\
		&=&\frac{\partial f}{\partial p_i}\frac{\partial g}{\partial q^i}-
		\frac{\partial f}{\partial q^i}\frac{\partial g}{\partial p_i}
		-\frac{\partial f}{\partial S}\left(p_i\frac{\partial g}{\partial p_i}\right)+\frac{\partial g}{\partial S}\left(p_i\frac{\partial f}{\partial p_i}\right)
	\end{eqnarray*}
		
	The main example of contact manifold for us will be  $T^*Q\times \R$, where  $Q$ is $n$-dimensional manifold, with  contact structure defined
	by
	\[
	\eta=pr_2^* (dS)-pr_1^* (\theta_Q)\equiv dS-\theta_Q
	\]
	where $pr_1: T^*Q\times \R\rightarrow T^*Q$ and $pr_2: T^*Q\times \R\rightarrow \R$  are the canonical projections and $\theta_Q$ is the Liouville 1-form on the cotangent bundle
	defined by
	\[
	\Theta_Q (X_{\mu_q})=\langle \mu_q, T_{\mu_{q}}\pi_Q X_{\mu_q}\rangle
	\]
		where $X_{\mu_q}\in T_{\mu_q} T^*Q$. 
Taking bundle  coordinates $(q^i,p_i)$ on $T^*Q$ we have that 
$\eta=dS-p_idq^i$.

On such a manifold we can define the bi-vector
\[
\Lambda_0=\Lambda+\sharp_{\Lambda}(dS)\wedge R
\]
which is Poisson, that is $[\Lambda_0, \Lambda_0]=0$. 
In coordinates, 
\[
\Lambda_0=\frac{\partial}{\partial p_i}\wedge \frac{\partial}{\partial q^i}
\]
is like the canonical Poisson bracket on $T^*Q$ but now applied to functions on $T^*Q\times \R$.  

Observe  that in this case the Cartan bracket can be rewritten in terms of the Poisson bracket  induced by  $\Lambda_0$ and an extra term that describe the thermodynamical behaviour. That is, 
\[
[f, g]=\{f, g\}_{\Lambda_0}-\frac{\partial f}{\partial S}\Delta g+\frac{\partial g}{\partial S}\Delta f
\]
where $\Delta=-\sharp_{\Lambda}(dS)$ is the Liouville vector field: 
\[
\Delta=p_i\frac{\partial}{\partial p_i}
\]		
We will denote by 
\[
\{
f, g
\}_{\Delta}=\frac{\partial g}{\partial S}\Delta f-\frac{\partial f}{\partial S}\Delta g
\]
	then the Cartan bracket is written as in the single generation formalism \cite{Ed-Ber, Ed-Ber-2}  as
	\begin{equation}\label{eee}
	[f, g]=\{f, g\}_{\Lambda_0}+\{
	f, g
	\}_{\Delta}
	\end{equation}
	
	Now, we will discuss some interesting properties of the qualitative behaviour of the evolution vector field $E_f$. In  \cite{BLMP} appears  a similar result for contact hamitonian vector fields (see also \cite{Godbillon1969}).

	\begin{proposition}\label{lie_dv_Ef_eta}
		We have that 
		$${\mathcal L}_{{\mathcal E}_f}\eta=-R(f)\eta +df\; .$$
	\end{proposition}
	\begin{proof}
		The proof is a trivial consequence of the properties of the Lie derivative and the properties of the Hamiltonian vector field (see \cite{marle}): 
	\begin{eqnarray*}
		{\mathcal L}_{{\mathcal E}_f}\eta&=&{\mathcal L}_{X_f+fR}\eta={\mathcal L}_{X_f}\eta +{\mathcal L}_{fR}\eta\\
		&=&-R(f)\eta+(i_R\eta)df=-R(f)\eta +df
		\end{eqnarray*}
    \end{proof}
	
	\begin{theorem}
		Let ${\mathcal L}_c(f)=f^{-1}(c)$ be a level set of $f: M\rightarrow \R$ where $c\in \R$. We assume that ${\mathcal L}_c(f)\not=0$ and $R(f)(x)\not= 0$ for all $x\in {\mathcal L}_c(f)$. Then
		\begin{enumerate}
		\item The 2-form $\omega_c\in \Omega^2({\mathcal L}_c(f))$ defined by
		\[
		\omega_c=-d i_{c}^*\eta
		\]
		is an exact symplectic structure. Here
		 $i_c: {\mathcal L}_c f\hookrightarrow M$ denotes  the canonical inclusion 
	\item If $\Delta_c$ is the Liouville vector field, that is, 
	\[
	i_{\Delta_c}\omega_c=i_{c}^*\eta
	\]
	then the restriction of ${\mathcal E}_f$ to ${\mathcal L}_c(f)$ verifies that
	$${\mathcal E}_{f}\big|_{{\mathcal L}_c (f)}=R(f)\big|_{{\mathcal L}_c(f)} \Delta_c$$
		\end{enumerate}
    \end{theorem}
    
    \begin{proof}
	The form $\omega_c$ is trivially closed. To see that it is a symplectic form, we just need to check that is non degenerate. Let $p \in {\mathcal L}_c (f)$. Notice that, at that point, $\omega_c = -d \eta\vert_{ T_p {\mathcal L}_c (f)}$. By the condition $R(f) \neq 0$, we have that $R_p$ (and, hence $\ker \eta = \hbox{span }\langle R \rangle$) is transverse to $T_p {\mathcal L}_c (f)$. But since $\eta_p \wedge d \eta_p^n \neq 0$,then $d \eta\vert_{ V}$ is non-degenerate for every subspace $V$ transverse to $\ker \eta$. Therefore, $\omega_c$ is also non-degenerated. 

	For the second part, we first remark that ${\mathcal E}_f(f) = 0$, hence ${(i_c)}_* {\mathcal E}_f = {\mathcal E}_f\vert_{{\mathcal L}_c (f)}$ is a well-defined vector field. By~Proposition~\ref{lie_dv_Ef_eta} and Cartan's identity
	\[
		i_{{\mathcal E}_f} d \eta = -R(f) \eta + df.
	\]
	Pulling back by $i_c$, we get
	\[
		i_{{(i_c)}_*{\mathcal E}_f} i_c^*d\eta = 
		-(R(f)\circ i_c) i_c^*\eta+ d i_c^*f = -(R(f)\circ i_c) i_c^*\eta,
	\]
	dividing by $-(R(f)\circ i_c) $,
	\[
		-i_{{(i_c)}_*{\mathcal E}_f/R(f)} i_c^*d\eta =i_{{(i_c)}_*{\mathcal E}_f/R(f)}\omega_c=i_{c}^*\eta.
	\]
	Thus, ${(i_c)}_*\left({\mathcal E}_f/R(f)\right) = \Delta_c$, as we wanted to show.
	\end{proof}
	Observe that since 	$${\mathcal E}_{f}\big|_{{\mathcal L}_c (f)}=R(f)\big|_{{\mathcal L}_c(f)} \Delta_c$$
	then the dynamics on each energy level is  like a Liouville dynamics after a time reparametrization $$dt = \frac{1}{R(f)}d\tau\; .$$
	
	\section{Simple mechanical systems with friction}\label{sec3}
	In this section, we will describe using the evolution vector field simple thermodynamic systems, that is systems  for which one scalar thermal variable (in our case the entropy)  and a finite set of mechanical variables (position and momenta) are enough to describe all the possible  states of the system. We assume that the system is adiabatically closed, that is, systems where there is not associated  transfer outside of work, matter  or heat. 
	That is, we consider adiabatically closed thermodynamic	systems.
	In this case, the thermodynamical simple systems are described by a Lagrangian  function:
	\[
	\begin{array}{rrcl}
	L:& TQ\times \R&\longrightarrow& \R\\
	   & (v_q, S)&\longmapsto& L(v_q, S)
	\end{array}
	\]
	where $Q$ is the configuration manifold describing the mechanical part of the thermodynamical system, $TQ$ the tangent bundle with canonical projection $\tau_Q: TQ\rightarrow Q$ given by $\tau_Q(v_q)=q$. 
	The entropy of the system is described by the real variable $S\in \R$. 
	If we consider coordinates $(q^i)$ on $Q$ and induced coordinates $(q^i, \dot{q}^i)$ on $TQ$ then $\tau_Q(q^i, \dot{q}^i)=(q^i)$.
	
	We will see that the Lagrangian function itself will produce a friction force satisfying naturally the two laws of thermodynamics. 
	
	We will assume that the Lagrangian system is regular, that is, the matrix
	\[
	(W_{ij})=\left(\frac{\partial^2 L}{\partial \dot{q}^i\partial \dot{q}^j}\right)
	\]
	is regular and the mapping $\F L: TQ\times \R\rightarrow T^*Q\times \R$ is a local diffeomorphism, where:
	\[
	 \F L(q^i, \dot{q}^i, S)=(q^i, \frac{\partial L}{\partial \dot{q}^i}, S)
	\]
	is the Legendre transform. Then, we may define a Hamiltonian function $H: T^*Q\times \R\rightarrow \R$ given by
	\[
	H(q^i, p_i, S)=p_i \dot{q}^i-L(q^i, \dot{q}^i, S)
	\]  
	where now the coordinates $\dot{q}^i$ are  implicitly  defined by the relations $p_j=\frac{\partial L}{\partial \dot{q}^j}(q^i, \dot{q}^i, S)$.

The equations of motion defined by the evolution vector field ${\mathcal E}_H$ are
\begin{eqnarray*}	
\frac{dq^i}{dt}&=&
	\frac{\partial H}{\partial p_i}\; ,\\
	\frac{dp_i}{dt}&=&
	-\frac{\partial H}{\partial p_i}-p_i\frac{\partial H}{\partial S}\\
	\frac{dS}{dt}&=&
	p_i\frac{\partial H}{\partial p_i}\; .
	\end{eqnarray*}

	The vector field ${\mathcal E}_H$ satisfies the following two  properties that are related with thermodynamical   systems  that conserves its energy, but redistributes it in an irreversible way, that property is collected by the variable $S$, the entropy of the system. 
	
\begin{proposition}\label{pro11}
	The integral curves of ${\mathcal E}_H$ satisfies the following properties: 
	\begin{enumerate}
		\item ${\mathcal E}_H(H)=0$, that is, $\frac{d H}{dt}=0$; 
		\item ${\mathcal E}_H(S)=\Delta (H)$, that is, $\frac{dS}{dt}=\Delta H$. 
	\end{enumerate}
	\end{proposition}
\begin{proof}
	Both are consequence of the definition of the evolution vector field ${\mathcal E}_H=\sharp_{\Lambda} (dH)$.
\end{proof}

Assume that the Hamiltonian $H$ is given by
\begin{equation}\label{hami}
H(q^i, p_i, S)=\frac{1}{2} g^{ij}p_ip_j+V(q, S)
\end{equation}
where $(g^{ij})$ is positive  semi-definite (for instance, it is associated to a Riemannian metric on $Q$). 
Then, the  vector field ${\mathcal E}_H$ describes a thermodynamical system with friction satisfying the first two laws of the thermodynamics: 
\begin{proposition}
	The integral curves of ${\mathcal E}_H$ satisfies the following properties: 
	\begin{enumerate}
		\item {\bf First law of Thermodynamics}:
		\begin{equation*}
			\frac{d H}{dt}=0 \quad \text{(preservation of the total energy);}
		\end{equation*} 
		\item {\bf Second law of Thermodynamics}:
		\begin{equation*}
		\frac{dS}{dt}=\Delta H\geq 0 \quad \text{(total entropy of an isolated system never decreases).}
		\end{equation*} 
	\end{enumerate}
\end{proposition}
\begin{proof}
	It is a direct consequence of Proposition \ref{pro11} and $\Delta H= p_ig^{ij}p_j\geq 0$.  
\end{proof}	

If we express the dynamics in terms of the brackets defined in (\ref{eee}) we have that 
	\begin{equation}\label{poi}
	\dot{f}=\{f, H\}_{T^*Q}+\{f, H\}_{\Delta}.
	\end{equation}
Obviously, $\{H, H\}_{T^*Q}=\{H, H\}_{\Delta}=0$  (first law) and 
$\{S, H\}_{T^*Q}=0$ and $\{S, H\}_{\Delta}=\Delta H\geq 0$ (second law). Observe that in Equation (\ref{poi}) both brackets are using the function $H$ as "generator". This is the reason that typically this formalism is known as {\sl single generator formalism} \cite{Ed-Ber}.

\begin{example}{\bf Linearly damped system}
{\rm 
	
 Consider a linearly damped system described by coordinates  $(q, p, S)$, where $q$ represents the position, $p$ the momentum of the particle and $S$ is the entropy of the surrounding thermal bath. 
The system is described by the  Hamiltonian 
\[
H(q, p, S)=\frac{p^2}{2m}+V(q)+\gamma S
\]
Therefore, the equations of motion for ${\mathcal E}_H$ are: 
\begin{eqnarray*}
	\dot{q}&=&\frac{p}{m}\\
	\dot{p}&=&-V'(q)-\gamma p\\
	\dot{S}&=& \frac{p^2}{m}
	\end{eqnarray*}
Obviously $\dot{H}=0$ and $\dot{S}\geq 0$. 

In the Lagrangian side we obtain the system: 
\begin{eqnarray*}
m\ddot{q}+\gamma \dot{q}+V'(q)&=&0\\
	\dot{S}&=& m\dot{q}^2
\end{eqnarray*}
observe that in this system the friction term is given by the 1-form $F_{fr}(q, \dot{q})=-\gamma \dot{q} dq$
Therefore the equation of temporal evolution of the entropy can be rewritten as follows
\[
\dot{S}=-\frac{1}{T} \langle F_{fr}(q, \dot{q}), \dot{q}\rangle
\]
where $T=\frac{\partial H}{\partial S}=-\frac{\partial L}{\partial S}=\gamma$ represents the temperature of the thermal bath. (see \cite{Gay-Balmaz2017, Gay-Balmaz2019}).

Obseve that the two brackets are: 
\begin{eqnarray*}
	\{f, g\}_{\Lambda_0}&=&\frac{\partial f}{\partial p}\frac{\partial g}{\partial q}-\frac{\partial g}{\partial p}\frac{\partial f}{\partial q}\\
		\{f, g\}_{\Delta}&=&p\frac{\partial g}{\partial S}\frac{\partial f}{\partial p}-p\frac{\partial f}{\partial S}\frac{\partial g}{\partial p}
\end{eqnarray*}
In particular
\begin{eqnarray*}
	\{H, g\}_{\Lambda_0}&=&\frac{p}{m}\frac{\partial g}{\partial q}-\frac{\partial g}{\partial p}V'(q)\\
	\{H, g\}_{\Delta}&=&\frac{p^2}{m}\frac{\partial g}{\partial S}-\gamma p\frac{\partial g}{\partial p}
\end{eqnarray*}
and
\[
{\mathcal E}_H(g)=\dot{g}=\{H, g\}_{\Lambda_0}+\{H, g\}_{\Delta}
\]
Therefore it is clear that $\{H, H\}_{\Lambda_0}=0$ and $\{H, H\}_{\Delta}=0$ (by skew-symmetry) and 
$\{H, S\}_{\Lambda_0}=0$ and 
$\{H, S\}_{\Delta}=\frac{p^2}{m}\geq 0$.
}
\end{example}
\section{Geometric integration of simple thermodynamical systems}\label{sec4}

\subsection{Integration based on discrete gradients}

Numerical methods for general thermodynamical systems are implemented usually  using the metriplectic formalism (see \cite{mielke, ignacio}), however in our case, for the examples that we are considering,   we can easily adapt the construction of discrete gradient methods to the bivector $\Lambda$.

For  simplicity, we will assume that $Q=\R^{N}$. Then the systems that we want to study are described by the  ODEs 
$$\dot{x}=(\sharp_{\Lambda})_x (\nabla H(x)),$$
with $x=(q^i, p_i, S)\in \mathbb{R}^{2n+1}$ and $\nabla H(x)\in\mathfrak{X}(Q)$ is the standard gradient in $\R^{N}$ with respect to the euclidean metric.

Using discretizations of the gradient $\nabla H(x)$ it is possible to define a class of integrators which preserve the first integral $H$ exactly.

\begin{definition}\label{def31}
	Let $H:\mathbb{R}^N\longrightarrow \mathbb{R}$ be a differentiable function. Then $\bar{\nabla}H:\mathbb{R}^{2N}\longrightarrow \mathbb{R}^N$ is a discrete gradient of $H$ if it is continuous and satisfies
	\begin{subequations}
		\label{discGrad}
		\begin{align}
		\bar{\nabla}H(x,x')^T(x'-x)&=H(x')-H(x)\, , \quad \, \mbox{ for all } x,x' \in\mathbb{R}^N  \, ,\label{discGradEn} \\
		\bar{\nabla}H(x,x)&=\nabla H(x)\, , \quad \quad \quad \quad \mbox{ for all } x \in\mathbb{R}^N  \, . \label{discGradCons}
		\end{align}
	\end{subequations}
\end{definition}

Some examples of discrete gradients are 
\begin{itemize}
	\item The {\bf mean value (or averaged) discrete gradient} given by
	\begin{equation}
	\label{AVF}
	\bar{\nabla}_{1}H(x,x'):=\int_0^1 \nabla H ((1-\xi)x+\xi x')d\xi \, , \quad \mbox{ for } x'\not=x \, .
	\end{equation}
	
	\item The {\bf midpoint (or Gonzalez) discrete gradient} given by
	\begin{align}
	\bar{\nabla}_{2}H(x,x')&:=\nabla H\left( \frac{1}{2}(x'+x)\right)\label{gonzalez}\\&+\frac{H(x')-H(x)
		-\nabla H\left( \frac{1}{2}(x'+x)\right)^T(x'-x)}{|x'-x|^2}(x'-x) \, , \nonumber
	\end{align}
for  $x'\not=x$. 
	\item The {\bf coordinate increment discrete gradient} where  each component given by
	\begin{equation}
	\label{itoAbe}
	\bar{\nabla}_{3}H(x,x')_i=\frac{H(x'_1,\ldots,x'_i,x_{i+1},\ldots,x_n)-H(x'_1,\ldots,x'_{i-1},x_{i},\ldots,x_n)}{x'_i-x_i}\, \nonumber
	\end{equation}
	$1\leq i \leq N$,  
	when $x_i'\not=x_i$, and $$\bar{\nabla}_{3}H(x,x')_i=\frac{\partial H}{\partial x_i}(x'_1,\ldots,x'_{i-1},x'_i=x_{i},x_{i+1},\ldots,x_n),$$ otherwise.
\end{itemize}

Once a discrete gradient $\bar{\nabla}H$ has been chosen, it is straightforward to define an energy-preserving integrator by, for instance, using the midpoint discrete gradient: 
\begin{equation}\label{discrete:gradient:integrator}
\frac{x_{k+1}-x_k}{h}=(\sharp_{\Lambda})_{(x_k+x_{k+1})/2}\bar{\nabla}_2H(x_k,x_{k+1}),
\end{equation}
where $\Lambda$ is the bivector associated to the canonical contact structure $\eta_{Q}$ of $Q=\R^{2n+1}$, given in local coordinates by \eqref{Lambda:coordinates}.

As in the continuous case, it is immediate to check that $H$ is exactly preserved using \eqref{discrete:gradient:integrator} and the skew-symmetry of $\Lambda$
\begin{equation*}
	\begin{split}
		H(x_{k+1})-H(x_k) & =\bar{\nabla}_2H(x_k,x_{k+1}')^T(x_{k+1}-x_k) \\
		& =h \Lambda (\bar{\nabla}_2H(x_k,x_{k+1}),  \bar{\nabla}_2H(x_k,x_{k+1}))=0.
	\end{split}
\end{equation*}

On the other hand, by \eqref{discrete:gradient:integrator} the entropy satisfies
$$
S_{k+1}-S_k=h\Lambda (\bar{\nabla}_2H(x_k,x_{k+1}), dS). 
$$
If $H$ is of the form (\ref{hami}) with $V$ a quadratic function then 
\[
H(x_{k+1})-H(x_k)=dH\left(\frac{x_k+x_{k+1}}{2}\right) (x_{k+1}-x_k).
\]
In fact this is a well-known property of quadratic functions. Hence, we must have
$$d H\left(\frac{x_k+x_{k+1}}{2}\right)=\bar{\nabla}_2 H(x_k,x_{k+1}),$$
so that
$$
S_{k+1}-S_k=h\Lambda \left(d H\left(\frac{x_k+x_{k+1}}{2}\right), dS\right)=h\frac{p_{k}^{i}+p_{k+1}^{i}}{2}\frac{\partial H}{\partial p^{i}}\left(\frac{x_k+x_{k+1}}{2}\right)\geq 0,$$
since by \eqref{Lambda:coordinates} we have that
\begin{equation*}
	\Lambda(dq^{i},dS)=0, \quad \Lambda(dp_{i},dS)=p_{i} \quad \text{and} \quad \Lambda(dS,dS)=0.
\end{equation*}

\begin{example}
	Consider the Hamiltonian function $H:T^{*}Q\rightarrow \R$ given by
	\begin{equation}\label{harmonic:osc}
	H(q,p,S)=\frac{p^{2}}{2}+\frac{q^{2}}{2}+\gamma S,
	\end{equation}
	where $Q=\R$, which is the Hamiltonian function associated with the damped harmonic oscillator.
	
	Now, if we may apply the midpoint discrete gradient and the associated integrator given by \eqref{discrete:gradient:integrator}, we obtain the following integrator
	\begin{equation}\label{DG2:harmonic}
	\begin{split}
	q_{1} = & \frac{2 \gamma h q_{0}-h^2 q_{0}+4 h p_{0}+4 q_{0}}{2 \gamma h+h^2+4} \\
	p_{1} = & -\frac{2 \gamma h p_{0}+h^2 p_{0}+4 h q_{0}-4 p_{0}}{2 \gamma h+h^2+4} \\
	S_{1} = & \frac{S_{0} h^4+(4 S_{0} \gamma+4 q_{0}^2) h^3+(4 S_{0} 	\gamma^2-16 p_{0} q_{0}+8 S_{0}) h^2}{(2 \gamma h+h^2+4)^2} \\
	& +\frac{(16 S_{0} \gamma+16 p_{0}^2) h+16 S_{0}}{(2 \gamma h+h^2+4)^2}.
	\end{split}
	\end{equation}
	
	Of course, using equations \eqref{DG2:harmonic} we obtain an integrator with constant energy and increasing entropy. In figures \ref{fig:test1} we can see that the qualitative behaviour of the integrator is fairly accurate, while in \ref{fig:test2} we see the entropy increases at the same rate as the exact one.
	
	\begin{figure}[htb!]
		\centering
		\includegraphics[width=0.7\linewidth]{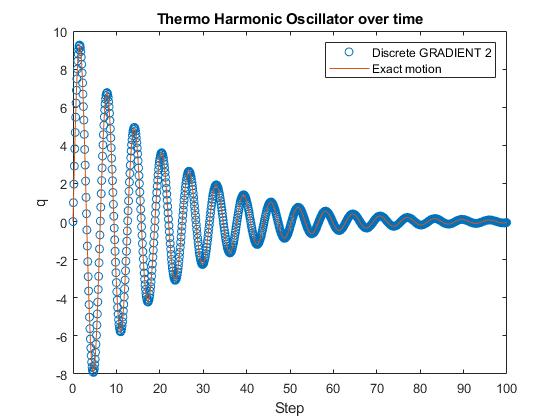}
		\caption{Trajectory of \eqref{DG2:harmonic}: the initial data are $q_{0}=0$, $p_{0}=10$ and $S_{0}=0$; the step is $h=0.1$ and $\gamma=0.1$. We plot the positions $q_{k}$ and compare the integrator with the integral curve of the evolution dynamics ${\mathcal E}_{H}$.}
		\label{fig:test1}
	\end{figure}
	
	\begin{figure}[htb!]
		\centering
		\includegraphics[width=0.7\linewidth]{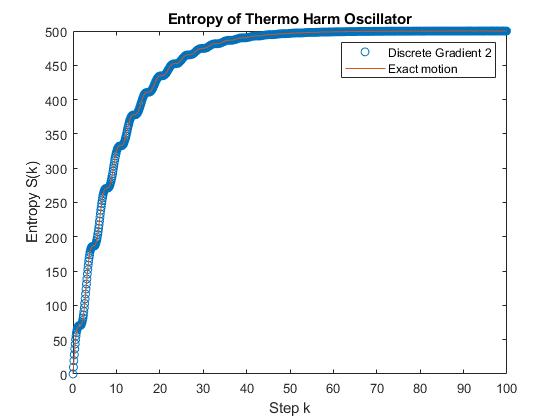}
		\caption{Error of \eqref{DG2:harmonic}: using the same initial data and settings from Figure \ref{fig:test1}, we plot the error with respect to the exact motion.}
		\label{fig:test2}
	\end{figure}
	
\end{example}

\section{Conclusions and future work}

We have shown the importance of the evolution or horizontal vector field to describe simple thermodynamical systems. We have proven that the restriction of this vector field to constant energy hypersurfaces is a time reparametrization of a Liouville vector field. Also, the relation with the single generation formalism of \cite{Ed-Ber} is elucidated and the construction of geometric integrators satisfying  the two laws of thermodynamics. 

Of course, our techniques are applied only to simple thermodynamical systems but we consider them to be the building blocks to model more evolved thermodynamical systems using interconnection of these simple systems as in  \cite{vander0}. We will study this framework in a future paper.  

Moreover, we will study the possibility of introducing the techniques developed in discrete mechanics, in particular, variational integrators, to numerically integrate the equations of the evolution vector field associated to a given Lagrangian function $L : TQ \times \R \longrightarrow \R$.
This would allow us to develop higher order  methods in a simple way as in \cite{marsden-west}.  In recent papers such as \cite{AMLL,VBS} a discrete Herglotz principle is introduced, allowing to obtain integrators for Lagrangian contact systems. We think that it is possible to adapt the previous constructions to the case of evolution vector fields. We will now develop some of the lines of this future research. 

\subsection{The geometric setting}

Let $L : TQ \times \R \longrightarrow \R$ be a regular Lagrangian function as in Section \ref{sec3} (see \cite{deLeon2019,deLeon2020}). As before, let us introduce coordinates on $TQ \times \R$, denoted by $(q^i, \dot{q}^i, S)$, where $(q^i)$ are coordinates in $Q$, $(q^i, \dot{q}^i)$ are the induced bundle coordinates in $TQ$	and $S$ is a global coordinate in $\R$.

Given a Lagrangian function $L$, using the canonical endomorphism ${\mathbf S}$ on $TQ$ locally defined by
$$	{\mathbf S}= d q^i \otimes \frac{\partial}{\partial \dot{q}^i},	$$
one can construct a 1-form $\lambda_L$ on $TQ\times \R$ given by
$$	\lambda_L = {\mathbf S}^* (dL)	$$	
where now ${\mathbf S}$ and ${\mathbf S}^*$ are the natural extensions of ${\mathbf S}$ and its adjoint operator ${\mathbf S}^*$ to $TQ \times \R$ \cite{deLeon1987}.

Therefore, we have that
$$	\lambda_L = \frac{\partial L}{\partial \dot{q}^i} \, dq^i.	$$

Now, the 1-form on $TQ\times \R$ given by $\eta_L = dS -\lambda_L$ or, in local coordinates, by	
$$	\eta_L = dS - \frac{\partial L}{\partial \dot{q}^i} \, dq^i	$$
is a contact form on $TQ \times \R$ if and only if $L$ is regular; indeed, if $L$ is regular, then we may prove that	
$	\eta_L \wedge (d\eta_L)^n \not= 0$,
and the converse is also true. 

The corresponding Reeb vector field is given in local coordinates by
$$	{\mathcal R}_L = \frac{\partial}{\partial S} - W^{ij} \frac{\partial^2 L}{\partial \dot{q}^j \partial S} \, \frac{\partial}{\partial \dot{q}^i} ,$$
where $(W^{ij})$ is the inverse matrix of the Hessian $(W_{ij})$.

The energy of the system is defined by 
$$E_L = \Delta (L) - L	$$
where $\Delta = \dot{q}^i \, \frac{\partial}{\partial \dot{q}^i}$ is the natural extension of the Liouville vector field on $TQ$ to $TQ \times \R$. Therefore, in local coordinates we have that	
$$E_L = \dot{q}^i \, \frac{\partial L}{\partial \dot{q}^i} - L.$$

Denote by $\flat_L : T(TQ \times \R) \longrightarrow T^* (TQ \times \R)$	the vector bundle isomorphism given by	
$$\flat_L (v) = i_v (d\eta_L) + (i_v \eta_L) \, \eta_L$$
where $\eta_L$ is the contact form on $TQ \times \R$ previously defined. We shall denote its inverse isomorphism by $\sharp_L = (\flat_L)^{-1}$.

Let ${\xi}_L$ be the unique vector field satisfying the equation	
\begin{equation}\label{clagrangian1}
	\flat_L ({\xi}_L) = dE_L - (\mathcal R_L E_L + E_L) \, \eta_L.
\end{equation}	



A direct computation from eq. (\ref{clagrangian1}) shows that if $(q^i(t), \dot{q}^i(t), S(t))$ is an integral curve of ${\xi}_L$, then it satisfies the generalized Euler-Lagrange equations considered by G. Herglotz in 1930:
\begin{equation}\label{clagrangian4}
	\begin{split}
		& \frac{d}{dt} \left(\frac{\partial L}{\partial \dot{q}^i}\right) - \frac{\partial L}{\partial q^i} = \frac{\partial L}{\partial \dot{q}^i} \frac{\partial L}{\partial S}\; ,\\
		& \dot{S}=L(q^i, \dot{q}^i, S)\; .
	\end{split}
\end{equation}

Now, given a regular Lagrangian function $L$, we may define the bi-vector $\Lambda_{L}$ on $TQ\times\R$ as in \eqref{Lambda:intrinsic} associated to the contact form $\eta_{L}$. That is, 
\begin{equation}\label{Lambda:intrinsic-1}
\Lambda_L(\alpha, \beta)=-d\eta_L(\flat_L^{-1}(\alpha), \flat_L^{-1}(\beta)), \qquad \alpha, \beta \in \Omega^1(TQ\times \R)\; .
\end{equation}

If $(q^i(t), \dot{q}^i(t), S(t))$ is an integral curve of the evolution vector field ${\mathcal E}_{L}$ associated to the contact form $\eta_{L}$ defined by
\begin{equation*}
	{\mathcal E}_{L}=\sharp_{\Lambda_{L}}(dE_{L}) \hbox{   or   }  	\flat_L ({\xi}_L) = dE_L - (\mathcal R_L E_L) \, \eta_L\; ,
\end{equation*}
then it satisfies the thermodynamical Herglotz equations
\begin{equation}\label{Herglotz:thermo}
	\begin{split}
		& \frac{d}{dt} \left(\frac{\partial L}{\partial \dot{q}^i}\right) - \frac{\partial L}{\partial q^i} = \frac{\partial L}{\partial \dot{q}^i} \frac{\partial L}{\partial S}. \\
		& \dot{S}=\dot{q}^i\frac{\partial L}{\partial \dot{q}^i}.
	\end{split}
\end{equation}

Moreover, if $H$ is the Hamiltonian function defined by $H=E_{L}\circ (\F L)^{-1}$, where $\F L:TQ\times \R\rightarrow T^{*}Q\times \R$ is the Legendre transform, then the evolution vector field ${\mathcal E}_{H}$ associated to $H$ is $\F L$-related to ${\mathcal E}_{L}$.


\subsection{Integration based on discrete Herglotz principle}

Now, we propose   to construct a numerical integrator for ${\mathcal E}_L$  based on a similar method to the discrete Herglotz principle \cite{AMLL, VBS}. 

Let $L_{d}:Q\times Q\times \R\rightarrow \R$ be a discrete Lagrangian function. Then a possible  integrator for the evolution dynamics is
\begin{equation}\label{DHE}
	D_{1}L_{d}(q_{1},q_{2},S_{1})+(1+D_{S}L_{d}((q_{1},q_{2},S_{1}))D_{2}L_{d}(q_{0},q_{1},S_{0})=0
\end{equation}
and the entropy is subjected to
\begin{equation}\label{Herglotz:entropy}
	S_{1}-S_{0}=(q_{1}-q_{0})D_{2}L_{d}(q_{0},q_{1},S_{0}).
\end{equation}

\begin{example}
	Consider again the Hamiltonian function \eqref{harmonic:osc} of the damped harmonic oscillator. Since $H$ is regular, we might consider the corresponding Lagrangian function $L:TQ\times \R\rightarrow \R$ given by
	\begin{equation*}
		L(q,\dot{q},S)=\frac{\dot{q}^{2}}{2}-\frac{q^{2}}{2}-\gamma S.
	\end{equation*}
	A standard discretization of this Lagrangian function is given by means of a quadrature rule like
	\begin{equation*}
		L_{d}(q_{0},q_{1},S_{0})=\frac{(q_{1}-q_{0})^{2}}{2h}-h\frac{(q_{1}+q_{0})^{2}}{8}-h \gamma S_{0}.
	\end{equation*}
	The discrete Herglotz equations \eqref{DHE} together with \eqref{Herglotz:entropy} give the explicit integrator
	\begin{equation}\label{Herglotz:harmonic}
		\begin{split}
			& q_{2} = \frac{\gamma h^3 q_{0}+\gamma h^3 q_{1}+4 \gamma h q_{0}-4 \gamma h q_{1}-h^2 q_{0}-2 h^2 q_{1}-4 q_{0}+8 q_{1}}{h^2+4} \\
			& S_{1} = S_{0} + \frac{(q_{1}-q_{0})^2}{h}-h \frac{q_{1}^{2}-q_{0}^{2}}{4}.
		\end{split}
	\end{equation}
	
	In Figures \ref{fig:test3} we plot the integrator given by equations \eqref{Herglotz:harmonic}. We see that the qualitative behaviour of the integrator is also quite good. In fact, an open question is whether the error can be improved by considering discrete Lagrangian functions approximating well enough the exact discrete Lagrangian function.
	
	As a last comment, the entropy for equations \eqref{Herglotz:harmonic} is increasing and the Hamiltonian oscillates before stabilizing around a constant value (cf. Fig \ref{fig:test5}).

	\begin{figure}[htb!]
		\centering
		\includegraphics[width=0.7\linewidth]{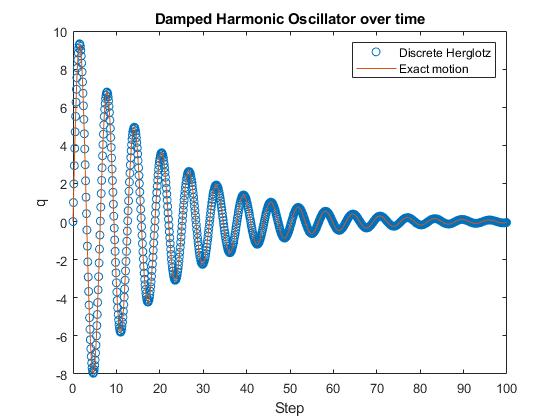}
		\caption{Trajectory of \eqref{Herglotz:harmonic}: the initial data are $q_{0}=0$, $q_{1}=1$ and $S_{0}=0$; the step is $h=0.1$ and $\gamma=0.1$. We plot the positions $q_{k}$ and compare the integrator with the integral curve of the evolution dynamics ${\mathcal E}_{L}$.}
		\label{fig:test3}
	\end{figure}
	
	\begin{figure}[htb!]
		\centering
		\includegraphics[width=0.7\linewidth]{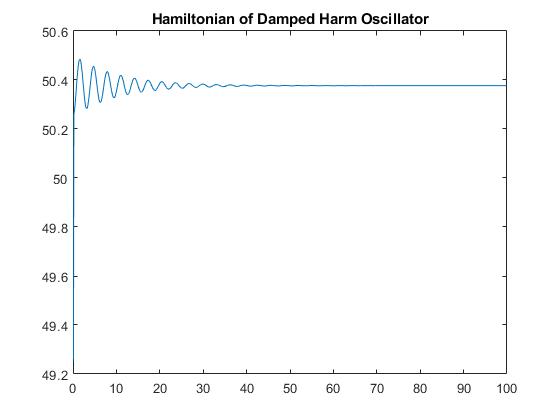}
		\caption{Hamiltonian of \eqref{Herglotz:harmonic}: using the same initial data and settings from Figure \ref{fig:test3}, we plot the Hamiltonian function along the iterations of the integrator.}
		\label{fig:test5}
	\end{figure}
	
\end{example}

\section*{Acknowledgments}

The authors  are supported  by  Ministerio de Ciencia e Innovaci\'on (Spain) under grants  MTM2016-76702-P and ``Severo Ochoa Programme for Centres of Excellence'' in R\&D (SEV-2015-0554). A.Simoes is supported by the FCT (Portugal) research fellowship SFRH/BD/129882/2017. Thank you very much  to E. Padr\'on for her helpful comments.

\bibliography{Thermo_Discrete}

\begin{thebibliography}{ASdLLMdD20}

\bibitem[ASdLLMdD20]{AMLL}
Anahory Anahory~Simoes, Manuel de~León, Manuel Lainz, and David Martín~de
  Diego.
\newblock On the geometry of discrete contact mechanics. arxiv:2003.11892
  [math.ph], 2020.

\bibitem[BdLMP20]{BLMP}
Alessandro Bravetti, Manuel de~León, Juan~Carlos Marrero, and Edith Padr\'on.
\newblock \textrm{Contact Hamiltonian systems, Reeb-Liouville dynamics and
  invariant measures}, \textit{Preprint}, 2020.

\bibitem[Bra17]{Bravetti2017}
Alessandro Bravetti.
\newblock Contact {{Hamiltonian Dynamics}}: {{The Concept}} and {{Its Use}}.
\newblock {\em Entropy}, 19(12):535, October 2017.

\bibitem[Bra18]{Bravetti2018}
Alessandro Bravetti.
\newblock Contact geometry and thermodynamics.
\newblock {\em Int. J. Geom. Methods Mod. Phys.}, 16(supp01):1940003, October
  2018.

\bibitem[dLLV19]{deLeon2018}
Manuel de~Le\'{o}n and Manuel Lainz~Valc\'{a}zar.
\newblock Contact {H}amiltonian systems.
\newblock {\em J. Math. Phys.}, 60(10):102902, 18, 2019.

\bibitem[dLV19]{deLeon2019}
Manuel {de Le{\'o}n} and Manuel Lainz~Valc{\'a}zar.
\newblock Singular {{Lagrangians}} and precontact {{Hamiltonian Systems}}.
\newblock {\em International Journal of Geometric Methods in Modern Physics
  (forthcoming)}, August 2019.

\bibitem[dLV20]{deLeon2020}
Manuel {de Le{\'o}n} and Manuel Lainz~Valc{\'a}zar.
\newblock Infinitesimal symmetries in contact hamiltonian systems.
\newblock {\em Journal of Geometry and Physics}, 2020.

\bibitem[dRR87]{deLeon1987}
Manuel {de Le{\'o}n} and Paulo R.~Rodrigues.
\newblock {\em Methods of Differential Geometry in Analytical Mechanics},
  volume 158.
\newblock {Elsevier}, {Amsterdam}, 1987.

\bibitem[EB91a]{Ed-Ber}
B.~J. Edwards and A.~N. Beris.
\newblock Noncanonical {P}oisson bracket for nonlinear elasticity with
  extensions to viscoelasticity.
\newblock {\em J. Phys. A}, 24(11):2461--2480, 1991.

\bibitem[EB91b]{Ed-Ber-2}
B.~J. Edwards and A.~N. Beris.
\newblock Noncanonical {P}oisson bracket for nonlinear elasticity with
  extensions to viscoelasticity.
\newblock {\em J. Phys. A}, 24(11):2461--2480, 1991.

\bibitem[EMvdS07]{vander0}
Damien Eberard, Bernhard Maschke, and Arjan~J. van~der Schaft.
\newblock On the interconnection structures of irreversible physical systems.
\newblock In {\em Lagrangian and {H}amiltonian methods for nonlinear control
  2006}, volume 366 of {\em Lect. Notes Control Inf. Sci.}, pages 209--220.
  Springer, Berlin, 2007.

\bibitem[God69]{Godbillon1969}
Claude Godbillon.
\newblock {\em {G{\'e}om{\'e}trie diff{\'e}rentielle et m{\'e}canique
  analytique}}.
\newblock {Hermann}, {Paris}, 1969.
\newblock OCLC: 1038025757.

\bibitem[Gon96]{Gonz}
O.~Gonzalez.
\newblock Time integration and discrete {H}amiltonian systems.
\newblock {\em J. Nonlinear Sci.}, 6(5):449--467, 1996.

\bibitem[GOR12]{ignacio}
Juan~Carlos Garc\'{\i}a~Orden and Ignacio Romero.
\newblock Energy-entropy-momentum integration of discrete thermo-visco-elastic
  dynamics.
\newblock {\em Eur. J. Mech. A Solids}, 32:76--87, 2012.

\bibitem[GP20]{GrPa}
Sergio Grillo and Edith Padr\'{o}n.
\newblock Extended {H}amilton-{J}acobi theory, contact manifolds, and
  integrability by quadratures.
\newblock {\em J. Math. Phys.}, 61(1):012901, 22, 2020.

\bibitem[GY17]{Gay-Balmaz2017}
Fran{\c c}ois {Gay-Balmaz} and Hiroaki Yoshimura.
\newblock A {{Lagrangian}} variational formulation for nonequilibrium
  thermodynamics. {{Part I}}: {{Discrete}} systems.
\newblock {\em Journal of Geometry and Physics}, 111:169--193, January 2017.

\bibitem[GY19]{Gay-Balmaz2019}
Fran{\c c}ois {Gay-Balmaz} and Hiroaki Yoshimura.
\newblock From {{Lagrangian Mechanics}} to {{Nonequilibrium Thermodynamics}}:
  {{A Variational Perspective}}.
\newblock {\em Entropy}, 21(1):8, January 2019.

\bibitem[IA88]{ITOH}
T.~Itoh and K.~Abe.
\newblock Hamiltonian-conserving discrete canonical equations based on
  variational difference quotients.
\newblock {\em J. Comput. Phys.}, 76(1):85--102, 1988.

\bibitem[Kau84]{Kaufman}
Allan~N. Kaufman.
\newblock Dissipative {H}amiltonian systems: a unifying principle.
\newblock {\em Phys. Lett. A}, 100(8):419--422, 1984.

\bibitem[LM87]{marle}
Paulette Libermann and Charles-Michel Marle.
\newblock {\em Symplectic geometry and analytical mechanics}, volume~35 of {\em
  Mathematics and its Applications}.
\newblock D. Reidel Publishing Co., Dordrecht, 1987.
\newblock Translated from the French by Bertram Eugene Schwarzbach.

\bibitem[Mie11]{mielke}
Alexander Mielke.
\newblock Formulation of thermoelastic dissipative material behavior using
  {GENERIC}.
\newblock {\em Contin. Mech. Thermodyn.}, 23(3):233--256, 2011.

\bibitem[MNCSS91]{Mrugala1991}
Ryszard Mrugala, James~D. Nulton, J.~Christian~Sch{\"o}n, and Peter Salamon.
\newblock Contact structure in thermodynamic theory.
\newblock {\em Reports on Mathematical Physics}, 29(1):109--121, February 1991.

\bibitem[Mor86]{morrison}
Philip~J. Morrison.
\newblock A paradigm for joined {H}amiltonian and dissipative systems.
\newblock volume~18, pages 410--419. 1986.
\newblock Solitons and coherent structures (Santa Barbara, Calif., 1985).

\bibitem[Mru93]{Mruga}
R.~Mrugala.
\newblock Continuous contact transformations in thermodynamics.
\newblock In {\em Proceedings of the {XXV} {S}ymposium on {M}athematical
  {P}hysics ({T}oru\'{n}, 1992)}, volume~33, pages 149--154, 1993.

\bibitem[MW01]{marsden-west}
J.~E. Marsden and M.~West.
\newblock Discrete mechanics and variational integrators.
\newblock {\em Acta Numer.}, 10:357--514, 2001.

\bibitem[QT96]{QT1996}
G.~R.~W. Quispel and G.~S. Turner.
\newblock Discrete gradient methods for solving {ODE}s numerically while
  preserving a first integral.
\newblock {\em J. Phys. A}, 29(13):L341--L349, 1996.

\bibitem[VBS19]{VBS}
Mats Vermeeren, Alessandro Bravetti, and Marcello Seri.
\newblock Contact variational integrators.
\newblock {\em J. Phys. A}, 52(44):445206, 28, 2019.

\bibitem[vdSM19]{vander}
Arjan van~der Schaft and Bernhard Maschke.
\newblock About some system-theoretic properties of port-thermodynamic systems.
\newblock In {\em Geometric science of information}, volume 11712 of {\em
  Lecture Notes in Comput. Sci.}, pages 228--238. Springer, Cham, 2019.

\end{thebibliography}
\bibliographystyle{alpha}
\end{document}